\documentclass[11pt]{article}
\usepackage{amsmath}
\usepackage{amsthm}
\usepackage{amssymb}
\usepackage{graphicx}
\usepackage{url}
\usepackage{authblk}

  \newcommand{\tr}{\operatorname{Tr}}

\newcommand{\be}{\begin{equation}}
\newcommand{\ee}{\end{equation}}
\newcommand{\BC}{\mathbb{C}}
\newcommand{\BR}{\mathbb{R}}

\def\vgap{\vskip 10pt}

\newtheorem{remark}{Remark}

\newtheorem{lemma}{Lemma}
\newtheorem{proposition}{Proposition}

\newcommand{\eref}[1]{Eq.~(\ref{#1})}
\def\vgap{\vskip 10pt}

\newcommand{\footer}[1]{{\def\thefootnote{}\footnotetext{#1}}}
\begin{document}
\hfill \today
\begin{center}
{\bf A Comment on ``{\bf\em On the Rotation Matrix in Minkowski Space-time}'' by \"Ozdemir and Erdo\u{g}du}\vspace{10 pt}

 \renewcommand{\thefootnote}{\fnsymbol{footnote}}
Arkadiusz Jadczyk\footnote[2]{Quantum Future Group, Inc., PO BOX 252, Almond, NC 28702, United States: {\tt ajadczyk@physics.org}} and
 \renewcommand{\thefootnote}{\fnsymbol{footnote}}
Jerzy Szulga\footnote[3]{Department of Mathematics and Statistics, Auburn University, Alabama 36849; email: {\tt szulgje@auburn.edu}}
\footer{\hspace{-20pt} \rule{2.1in}{0.5pt}}
\footer{{\bf 2013 AMS Subject Classification}: 15A90, 15A57, 17B81, also 22E70}
\footer{{\bf Key words and phrases}: {\em generalized Euler-Rodrigues formula, Minkowski space, Lorentz group,  $\mathrm{SL(2,\BC)}$, $\mathrm{SO_0(3,1)}$}}
\footer{\hspace{-20pt} \rule{2.1in}{0.5pt}}
\footer{~}
\vgap

\end{center}
 \renewcommand{\thefootnote}{\arabic{footnote}}
 \vgap

\begin{quote}
{\bf Abstract}. {\small We comment on the article  by M. \"Ozdemir and M. Erdo\u{g}du \cite{ozdemir2014}. We indicate that the exponential map onto the Lorentz group can be obtained in two elementary ways. The first way utilizes a commutative algebra involving a conjugate of a semi-skew-symmetric matrix, and the second way is based on the classical epimorphism from $\mathrm{SL(2,\BC)}$ onto $\mathrm{SO_0(3,1)}$.}
\end{quote}
\vgap\vgap

\section{Introduction\label{sec:intro}}
The classical Euler-Rodrigues formula gives the explicit form of the rotation matrix in $\BR^3$ in terms of the rotation axis and the rotation angle (see e.g \cite{palais2007} for a pedagogical introduction by B. Palais and R. Palais). It can be also interpreted as the explicit formula for the exponential of a $3\times 3$ skew-symmetric matrix. Various generalizations to other dimensions have been studied (see eg. \cite{bernstein1993}. In physics  a generalization to the Minkowski space especially matters. This problem has been mentioned by J. Gallier in his lecture notes on Lie groups \cite{gallier2005}, sending the reader to the Ph. D. thesis of C. M. Geyer \cite{geyer}. Geyer indeed has provided such a derivation but not quite optimal nor complete. In a recent paper E. Minguzzi \cite{minguzzi2013}  classified standard forms of generators and provided the formula for each class separately\footnote{An  introduction to such classification can be found in \cite{hall2004,naber2012}}. The present note emerged as an alternative method to one proposed recently by  M. Ozdemir and M. Erdo\u{g}du \cite{ozdemir2014}. Our derivation is based on simple algebraic properties of the algebra constructed from the generator and its dual.  The singular case of the generator with the quadruple null eigenvalues  is also discussed in our note.
\section{Minkowski space generalization of\\ Euler-Rodrigues
formula\label{sec:gen}}\index{Minkowski!space}
Let $F$ be a generator of a one--parameter subgroup of the Lorentz group. We write $F$ in the following general form that resembles the form of the electromagnetic field mixed tensor expressed ${F^\mu}_\nu$ in terms of the electric and magnetic field\index{magnetic field} vectors $\mathbf{e}$ and $\mathbf{b} :$
\be F= \left(\begin{smallmatrix}
0 & b_3 & -b_2 & e_1 \\
 -b_3 & 0 & b_1 & e_2 \\
 b_2 & -b_1 & 0 & e_3 \\
 e_1 & e_2 & e_3 & 0
\label{eq:F}\end{smallmatrix}\right).
\ee
The ``dual'' matrix, denoted by $\tilde{F}$, is obtained from $F$ by a ``dual rotation", that is, by replacing $\mathbf{e}\rightarrow \mathbf{b},\, \mathbf{b}\rightarrow -\mathbf{e} :$
\be \tilde{F}= \left(\begin{smallmatrix}
0 & -e_3 & e_2 & b_1 \\
 e_3 & 0 & -e_1 & b_2 \\
-e_2 & e_1 & 0 & b_3 \\
 b_1 & b_2 & b_3 & 0
\label{eq:Fstar}\end{smallmatrix}\right)
\ee
We also introduce real numbers $u$ and $v$ defined as
\be u=\frac{1}{4}\tr F\tilde{F}=\mathbf{e}\cdot\mathbf{b},\quad v=\frac{1}{4}\tr\,F^2=\frac{1}{2}(\mathbf{e}^2-\mathbf{b}^2).
\label{eq:uv}\ee
The characteristic polynomials for $F$ and $\tilde{F}$ can now be expressed in terms of $u$ and $v$
\be \det (F-\lambda I)=\lambda^4-2v\lambda^2-u^2,\quad \det (\tilde{F}-\lambda I)=\lambda^4+2v\lambda^2-u^2.\ee
Let $\sigma$ and $\theta$ be defined as
\be \sigma=\sqrt{\sqrt{u^2+v^2}+v},\quad \theta=\mbox{sgn}(u)\sqrt{\sqrt{u^2+v^2}-v}.\label{eq:theta}\ee
It is clear that $\sigma$ is nonnegative and $\theta$ has the sign of $u,$ where the $\mbox{sgn}$ is defined to be right continuous, that is $\mbox{sgn}(0)=1.$ The eigenvalues of $F$ are $\pm\sigma$ and $\pm i\theta$ while the eigenvalues of $\tilde{F}$ are $\pm \theta$ and $\pm i\sigma.$
The following identities follow directly from the definitions:
\be v=\frac{\sigma^2-\theta^2}{2},\quad u= \sigma\theta.\label{eq:uvab}\ee
Let us define $T$ as follows:
\be F^2+\tilde{F}^2=T.\label{eq:T}\ee
\begin{lemma}
The matrices $F,\tilde{F},T$ commute. Moreover, the following identities hold:
\be F\tilde{F}=\tilde{F} F=\sigma\theta\,I\,\label{eq:i1}\ee
\be F^2-\tilde{F}^2=(\sigma^2-\theta^2)\,I\label{eq:i2}\ee
\be F^3 =(\sigma^2-\theta^2)F+\sigma\theta\tilde{F},\label{eq:i3}\ee
\be F^2=\frac{T+(\sigma^2-\theta^2)I}{2},\label{eq:i4}\ee
\be\tilde{F}^2=\frac{T-(\sigma^2-\theta^2)I}{2},\label{eq:i5}\ee
\be  FT=(\sigma^2-\theta^2)F+2\sigma\theta\tilde{F},\label{eq:i6}\ee
\be \tilde{F} T=2\sigma\theta F-(\sigma^2-\theta^2)\tilde{F},\label{eq:i7}\ee
\be T^2=(\sigma^2+\theta^2)^2I.\label{eq:i8}\ee
\label{lem:lem1}
\end{lemma}
\begin{proof}
Equations (\ref{eq:i1}-\ref{eq:i3}) follow by a routine matrix algebra  from the identities $\sigma\theta=\mathbf{e}\cdot\mathbf{b}$ and $\sigma^2-\theta^2=\mathbf{e}^2-\mathbf{b}^2.$
\eref{eq:i4} (resp. (\ref{eq:i5})) follows by adding (resp. subtracting)  \eref{eq:T} and \eref{eq:i2}. \eref{eq:i6} results in a smilar way.
In order to show \eref{eq:i7} we first multiply \eref{eq:i2} by $\tilde{F}$ and use \eref{eq:i1}.
Finally, \eref{eq:i8} can be derived from \eref{eq:i6} multiplied by $F$, and by using (\ref{eq:i1}) and (\ref{eq:i2}).
\end{proof}
\begin{proposition}[Generalized Euler-Rodrigues formula]\label{app:propexp}
\index{Euler-Rodrigues formula}Assume that $\sigma^2+\theta^2 > 0.$ Then the following general formula holds:
\begin{align}
\exp(Ft)&=\frac{\cosh(t\sigma)+\cos(t\theta)}{2}\,I
+\frac{\sigma\sinh(t\sigma)+\theta\sin(t\theta)}{\sigma^2+\theta^2}\,F&\notag\\
&+\frac{\theta\sinh(t\sigma)-\sigma\sin(t\theta)}{\sigma^2+\theta^2}\tilde{F}
+\frac{\cosh(t\sigma)-\cos(t\theta)}{2(\sigma^2+\theta^2)}\,T.\label{eq:expft}
\end{align}
Equivalently, using $F^2$ instead of $T$
\begin{align}
\exp(Ft)&=\frac{\theta^2\cosh(t\sigma)+\sigma^2\cos(t\theta)}{\sigma^2+\theta^2}\,I
+\frac{\sigma\sinh(t\sigma)+\theta\sin(t\theta)}{\sigma^2+\theta^2}\,F&\notag\\
&+\frac{\theta\sinh(t\sigma)-\sigma\sin(t\theta)}{\sigma^2+\theta^2}\tilde{F}
+\frac{\cosh(t\sigma)-\cos(t\theta)}{\sigma^2+\theta^2}\,F^2.\label{eq:expft2}
\end{align}
If $\sigma=\theta=0,$ then
\be \exp(tF)=I+tF+\frac{t^2}{4}T=I+tF+\frac{t^2}{2}F^2.\label{eq:a0b0}\ee
\end{proposition}
\begin{proof}
In the proof we use the following theorem about generators of one--parameter matrix groups: If $\gamma(t)$ is a one-parameter group of matrices, then $\gamma(t)=\exp(Xt),$ where $X=\gamma'(0).$ \footnote{The proof of this classical theorem can be found, for instance, in `An introduction to matrix groups and their applications` by Andrew Baker, Springer 2002, Theorem 2.17, also available online, the same title and author, Theorem 2.5: \url{http://www.maths.gla.ac.uk/~ajb/dvi-ps/lie-bern.pdf}}
We consider first the case of at least one of the numbers $\sigma,\theta$ being nonzero, i.e. $\sigma^2+\theta^2>0.$

Let $L(t)$ denote the right hand side of \eref{eq:expft}. Immediately, $L(0)=I, L'(0)=F.$ We aim to show that $L(t)$ is a one-parameter matrix group, i.e., that
\be L(t+s)=L(t)L(s).\label{eq:gp}\ee
The proof is somewhat tedious but straightforward. We write $L(t+s)$ and expand the functions $\sin(x+y),$ $\cos(x+y),$ $\sinh(x+y),$ $\cosh(x+y)$ in terms of products of functions of the corresponding arguments $x,y$ (i.e., products of $t,s$ and $\sigma,\theta$).  This way we get a long expression with coefficients at the matrices $I,F,\tilde{F},T.$

On the other hand, we multiply $L(t)L(s)$ and obtain coefficients in front of the products of $I,F,\tilde{F},T.$ All of these products can be reduced to $I,F,\tilde{F},T$ using Lemma \ref{lem:lem1}. Comparing the coefficients in front $I,F,\tilde{F},T$ establishes the result.

Suppose now that $\sigma=\theta=0.$ Then, from Lemma \ref{lem:lem1}, we have that
\be FT=T^2=0,F^2=T/2\ee The group property of $L(t)$ given by \eref{eq:gp} follows then by the following observation:
\begin{align}
L(t)L(s)&=(I+tF+\frac{t^2}{4}T) (I+sF+\frac{s^2}{4}T)&\notag\\
&=I+sF+\frac{s^2}{4}T+tF+\frac{ts}{2}T+\frac{t^2}{4}T&\\
&=I+(s+t)F+\frac{1}{4}(s+t)^2T.
\notag\end{align}
On the other hand $L(0)=I,L'(0)=F,$ which completes the proof. \eref{eq:expft2} follows from \eref{eq:expft} and \eref{eq:i4}.
\end{proof}
\begin{remark}
Since we are dealing with commuting matrices, the problem reduces to a simple commutative symbolic algebra. It can be handled more efficiently by an adequate software, capable of commutative symbolic operations.
\end{remark}
\subsection{Alternative derivation via $\mathrm{SL(2,\BC)}$}
With the four Hermitian matrices\footnote{By abuse of notation $\sigma_\mu$ and $\sigma^\mu$ constitute exactly the same set matrices. Their components are $\sigma_{\mu AB}$ and $\sigma^{\mu AB},$ $(\mu=1,...,4)$, $(A,B=1,2).$}\index{hermitian:matrix} $\sigma_\mu=\sigma^\mu,$
\be \sigma_1=\left(\begin{smallmatrix}0&1\\1&0\end{smallmatrix}\right),\,
\sigma_2=\left(\begin{smallmatrix}0&i\\-i&0\end{smallmatrix}\right),\,
\sigma_3=\left(\begin{smallmatrix}0&0\\0&-1\end{smallmatrix}\right),\,
\sigma_4=\left(\begin{smallmatrix}1&0\\0&1\end{smallmatrix}\right),\ee
\index{spin!matrices} the group homomorphism $A\mapsto \Lambda(A)$ from the group of unimodular matrices $\mathrm{SL(2\BC)}$  onto the connected component of the identity $\mathrm{SO(3,1)}_0$ of the homogenous Lorentz group\index{Lorentz group!restricted} is given by
\be {\Lambda(A)^\mu}_\nu=\frac{1}{2}\tr (A\sigma^\mu A^\dagger\sigma_\nu),\,(\mu,\nu=1,...,4).\label{app:lmunu}\ee
The completeness relations for $\sigma$ matrices
\be \sum_{\mu=1}^4\sigma^{\mu AB}\sigma_{\mu CD}=2\delta^{A}_D\delta^{B}_C,\, (A,B,C,D=1,2)\ee
entail \be\tr(\Lambda(A))=|\tr(A)|^2.\ee
Taking the derivative of \eref{app:lmunu} we arrive at the linear relation (isomorphism) between infinitesimal generators $f$ (traceless $2\times 2$ complex matrices) from the Lie algebra $\mathrm{sl(2,\BC)}$ to the Lie algebra of elements $F$ in $\mathrm{so(3,1)}$:
\be {F^\mu}_\nu=\frac{1}{2}\tr(f\sigma^\mu\sigma_\nu+\sigma^\mu f^\dagger\sigma_\nu).\ee
With $f$ defined by \be f\stackrel{\text{def}}{=}\frac{1}{2}\sum_{i=i}^3(e_i+ib_i)\sigma_i,\ee
we arrive at $F$ given by (\ref{eq:F}), while $\tilde{f}\stackrel{\text{def}}{=}-if$ gives $\tilde{F}.$ The characteristic polynomial for $f$,
$\det(f-\lambda\,I)=\lambda^2-\frac{1}{2}(v+iu),$
entails two roots $\pm \omega.$ There is a simple relation between $\omega$ and $\sigma,\theta$:
$\omega=\frac{1}{2}(\sigma+i\theta).$
Every $2\times 2$ complex matrix $X$ determines a vector in the complex Minkowski space with complex coordinates $x^\mu=\tr(\sigma^\mu X)/2.$ There are two scalar products in this space: $(x,y)=x^TJy$ and $\{x,y\}=x^\dagger Jy.$ The first one is bilinear, while the second one is Hermitian. Both are $SO(3,1)$ invariant. $X$ is Hermitian if and only if $x^\mu$ are real, moreover $\tr(X^\dagger\epsilon Y\epsilon)/2=\{x,y\}$ and $\tr(X^T\epsilon X\epsilon)/2=\det(X)=(x,y),$ where $\epsilon=\left(\begin{smallmatrix}0&1\\-1&0\end{smallmatrix}\right).$ If $\xi_\pm$ are eigenvectors\index{eigenvector} of $f$ belonging to eigenvalues $\pm\omega\neq 0,$ and if $X_\pm=\xi_\pm\otimes\xi_\pm^\dagger$, then $x_\pm$ are real isotropic (i.e. $(x,y)=\{x,y\}=0$) eigenvectors of $F$ corresponding to real eigenvalues $\pm 2\Re(\omega)$. Vectors $y_\pm$ corresponding to $\xi_+\otimes\xi_-^\dagger$ and $\xi_-\otimes\xi_+^\dagger$ are Hermitian space-like (we have $\{y,y\}=2(||\xi_+||^2||\xi_-||^2-|\xi_+^\dagger\xi_-|^2)>0$), bilinear isotropic (i.e. $(y_\pm,y_\pm)=0),$ and $J$-orthogonal to $x_\pm,$ resp. They are eigenvectors of $F$ corresponding to imaginary eigenvalues\index{eigenvalue} $\pm 2\Im(\omega)$. Since $f^2=(v+iu)I/2=\omega^2,$ $\exp(tf)$ is easily computed
\be e^{tf}= \cosh(\omega t)\,I+\frac{\sinh \omega t}{\omega}\,f,\label{app:expf}\ee
where it does not matter which of the two possible signs of $\omega$ is chosen. If $\omega=0,$ then $f$ has just one eigenvector $\xi,$ vector $x,$ corresponding to $\xi$ is real isotropic, and F annihilates 2-dimensional plane in the three dimensional hyperplane orthogonal to $x.$ Moreover, when $\omega=0,$ which happens if and only if $f^2=0,$ we get instantly
\be e^{tf}= I+tf,\ee
which can be also obtained by taking the limit of $\omega\rightarrow 0$ in \eref{app:expf}. We can now expand the functions $\cos(t\omega),\sinh(t\omega)$ of the complex argument $t\omega=t(\sigma+i\theta)$ and use  \eref{app:lmunu} to recover the results of Proposition \ref{app:propexp}  by straightforward though somewhat lengthy calculations.

\end{document}